\documentclass[smallextended,final]{svjour3}     
\usepackage{mathrsfs}
\usepackage{amsfonts}
\usepackage{tikz}
\usepackage{mathptmx}
\usepackage{amssymb,bm}
\usepackage{amsmath}
\usepackage{bbding}
\usepackage{float}
\usepackage{fancyhdr}
\usepackage{geometry}
 \allowdisplaybreaks
\spnewtheorem{const}{Construction}{\bf}{\emph{\it}}
\newcommand{\Tr}{{\rm{Tr}}}
\newcommand{\tr}{{\rm{tr}}}

\newcommand{\rad}{\hbox{\rm{rad}}}

\newcommand{\qedd}{\hspace*{\fill}$\Box$\medskip}   

\begin{document}
\large


\title{New constructions of quaternary bent
functions}

\author{Baofeng Wu\and Dongdai Lin}

\titlerunning{New constructions of quaternary bent
functions}

\institute{Baofeng Wu \at
              State Key Laboratory of Information Security, Institute of Information Engineering, Chinese Academy of Sciences, Beijing 100093, China \\
\email{wubaofeng@iie.ac.cn}\and Dongdai Lin \at
              State Key Laboratory of Information Security, Institute of Information Engineering, Chinese Academy of Sciences, Beijing 100093, China \\
\email{ddlin@iie.ac.cn}}

\date{Received: date / Accepted: date}
\maketitle

\begin{abstract}
In this paper, a new construction of quaternary bent functions from
quaternary  quadratic forms over Galois rings of characteristic 4 is
proposed. Based on this construction, several new classes of
quaternary bent functions are obtained, and as a consequence,
several new classes of quadratic binary bent and semi-bent functions
in polynomial forms are derived. This work generalizes the recent
work of N. Li, X. Tang and T. Helleseth.


\end{abstract}



\section{Introduction}\label{secintro}

Let $\mathbb{F}_{2^n}$ be the finite field with $2^n$ elements,
where $n$ is a positive integer. Any map from $\mathbb{F}_{2^n}$ to
the integer residue ring $\mathbb{Z}_q$ is called an $n$-variable
$\mathbb{Z}_q$-valued Boolean function. Particularly,
$\mathbb{Z}_2$-valued Boolean functions are just the usual binary
Boolean functions and $\mathbb{Z}_4$-valued Boolean functions are
also  known as quaternary Boolean functions.

For an $n$-variable $\mathbb{Z}_q$-valued Boolean function $f$, its
Fourier transform at any $a\in\mathbb{F}_{2^n}$ is defined as
\[\hat{f}(a)=\sum_{x\in\mathbb{F}_{2^n}}\zeta_q^{f(x)}(-1)^{\tr_1^n(ax)},\]
where $\zeta_q$ is a $q$-th complex primitive root of unity and
``$\tr^n_1(\cdot)$" denotes the trace function from
$\mathbb{F}_{2^n}$ to $\mathbb{F}_{2}$, i.e.
$\tr(x)=\sum_{i=0}^{n-1}x^{2^i}$. $f$ is called a
$\mathbb{Z}_q$-valued (generalized) bent function if its Fourier
transform has a constant magnitude, or more precisely,
$|\hat{f}(a)|=2^{n/2}$ for any $a\in\mathbb{F}_{2^n}$. In the case
$q=2$, the Fourier transform of $f$ is always known as its Walsh
transform and $f$ is just the so-called bent function if it is a
$\mathbb{Z}_2$-valued bent function.

$\mathbb{Z}_q$-valued bent functions were introduced by Schmidt in
\cite{schmidt1} when seeking for good codes for multicode
code-division multiple access (MC-CDMA) systems. In fact, a
$\mathbb{Z}_q$-valued bent function corresponds to a code that can
reduce the peak-to-average power ratio (PAPR) \cite{jones} in such
systems to the lowest possible value (called a constant-amplitude
code). As a result, constructions of $\mathbb{Z}_q$-valued bent
functions will promise useful  objects in communication systems.

Generally speaking, $q$ is often chosen to be a power of 2 in
applications and the simplest case is $q=2$. But there are a main
drawback with binary bent functions that they only exist for even
number of variables. To avoid this drawback, a lot of attention has
been paid to quaternary bent functions as they exist for both even
and odd number of variables, and constructions of them have been
extensively studied. The main technique to construct quaternary bent
functions is to construct certain trace forms of Galois rings of
characteristic 4 and consider their limitations to the Teichm\"uller
sets since the Teichm\"uller sets are isomorphic to the finite
fields under certain multiplications and additions. For example, in
his Ph.D thesis \cite{schmidt2}, Schmidt considered quaternary
Boolean functions of the form
\begin{equation*}
Q(x)=\epsilon+\Tr_1^n(ax+2bx^3),~x\in\mathbb{T},
\end{equation*}
where $\epsilon\in\mathbb{Z}_4$, $a,~b\in GR(4,n)$ and
``$\Tr^n_1(\cdot)$" represents the trace function from $GR(4,n)$ to
$\mathbb{Z}_4$ (we denote the Galois ring with $4^n$ elements by
$GR(4,n)$ and
 its Teichm\"uller set by $\mathbb{T}$). He reduced the
conditions under which $Q$ was  a quaternary bent function to the
existence of roots of certain cubic equations over
$\mathbb{F}_{2^n}$. Very recently, Li et al. studied quaternary
Boolean functions of the form
\begin{equation*}
Q(x)=\Tr^n_1\left(x+2\sum_{i=1}^{\lfloor\frac{n-1}{2}\rfloor}c_ix^{1+2^{ki}}\right),~x\in\mathbb{T},
\end{equation*}
where $k$ is any positive integer and $c_i\in\mathbb{Z}_2$, $1\leq
i\leq \lfloor\frac{n-1}{2}\rfloor$ \cite{li}. It was proved that $Q$
is quaternary bent if and only if $\gcd(c(x^k),x^n-1)=1$ for
\[c(x)=1+\sum_{i=1}^{\lfloor\frac{n-1}{2}\rfloor}c_i(x^i+x^{n-i})\in\mathbb{F}_2[x].\]
Moreover, several classes of $c(x)$ satisfying such a condition were
constructed and thus several new classes of quaternary bent function
were obtained.

In this paper, we devote to generalizing Li et al.'s work. We
consider quaternary Boolean functions of the form
\begin{equation*}
Q(x)=\Tr^n_1\left(\alpha
x+2\sum_{i=1}^{\lfloor\frac{m-1}{2}\rfloor}c_i\beta
x^{1+2^{eki}}\right),~x\in\mathbb{T},
\end{equation*}
where $n=em$, $\alpha,~\beta\in\mathbb{T}$ and $c_i\in\mathbb{Z}_2$,
$1\leq i\leq \lfloor\frac{m-1}{2}\rfloor$. For some special choices
of $\alpha$ and $\beta$, we can derive the conditions for $Q$ to be
bent. Furthermore, we explicitly construct several classes of
coefficients sets $\{c_i\}$ meeting this condition, some of which
can be implied by Li et al.'s constructions.

On the other hand, by virtue of the connections between quaternary
bent functions  and binary bent and semi-bent functions deduced by
St\u{a}nic\u{a} et al. \cite{stanica}, we further derive new classes
of binary bent or semi-bent functions, respectively, of the form
\[f_Q(x)=p(\alpha x)+\sum_{i=1}^{\lfloor\frac{m-1}{2}\rfloor}c_i\tr_1^n(\beta x^{1+2^{eki}}),~x\in\mathbb{F}_{2^n},\]
according as $n$ is even or odd, respectively, where
\begin{equation}\label{eqpx}
p(x)=\left\{\begin{array}{ll}
\sum\limits_{i=1}^{\frac{n}{2}-1}\tr_1^n\left(x^{1+2^i}\right)+\tr_1^{n/2}\left(x^{1+2^{n/2}}\right)&~~\text{if}~n~\text{is~even},\\[.1cm]
\sum\limits_{i=1}^{\frac{n-1}{2}}\tr_1^n\left(x^{1+2^i}\right)&~~\text{if}~n~\text{is~odd}.
\end{array}
\right.
\end{equation}
 They are all quadratic bent or semi-bent functions in
polynomial  forms \cite{charpin}.

The rest of the paper is organized as follows. In Section
\ref{secpre}, we recall some necessary backgrounds on bent
functions, Galois rings and quadratic forms over them. In Section
\ref{secZ4bent}, our constructions of quaternary bent functions are
proposed, based on which some new classes of bent and semi-bent
functions are derived in Section \ref{secbent}. Concluding remarks
are given in Section \ref{secconclud}.

\large
\section{Preliminaries}\label{secpre}

\subsection{Quaternary bent functions and binary (semi-)bent functions}
Recall that an $n$-variable quaternary Boolean function $f$ is
called bent if $|\hat{f}(a)|=2^{n/2}$ for any $a\in\mathbb{F}_{2^n}$
where
\[\hat{f}(a)=\sum_{x\in\mathbb{F}_{2^n}}i^{f(x)}(-1)^{\tr_1^n(ax)},\]
$i=\sqrt{-1}$. An $n$-variable binary Boolean function $g$ is called
bent \cite{rothaus} if $|\hat{g}(a)|=2^{n/2}$ for any
$a\in\mathbb{F}_{2^n}$ where
\[\hat{g}(a)=\sum_{x\in\mathbb{F}_{2^n}}(-1)^{g(x)+\tr_1^n(ax)}.\]
Besides, $g$ is called semi-bent if
$|\hat{g}(a)|\in\{0,2^{\lfloor(n+2)/2\rfloor}\}$ \cite{chee}. It is
easy to derive that binary bent functions only exist for even $n$,
while quaternary bent functions exist for both even and odd $n$.

As every element of $\mathbb{Z}_4$ has a 2-adic expansion, there
exist two binary Boolean functions $f_0$ and $f_1$ such that
$f=f_0+2f_1$ for any quaternary Boolean function $f$. Under this
representation, connections between quaternary bent functions and
binary bent and semi-bent functions are obtained by St\u{a}nic\u{a}
et al. recently.

\begin{theorem}[\cite{stanica}]\label{connection}
Let $f(x)$ be an $n$-variable quaternary Boolean function and
$f(x)=f_0(x)+2f_1(x)$ be its 2-adic expansion. Denote
$\phi(f)(y,z)=f_0(y)z+f_1(y)$ with $y\in\mathbb{F}_{2^n}$ and
$z\in\mathbb{F}_{2}$, which can be viewed as an $(n+1)$-variable
Boolean function over $\mathbb{F}_{2^n}\times\mathbb{F}_{2}$. Then

\noindent (1)
$|\hat{f}(a)|^2=\left(|\hat{f_0}(a)|^2+|\hat{f_1}(a)|^2\right)/2$
for any $a\in\mathbb{F}_{2^n}$;

\noindent (2) $f(x)$ is quaternary bent if and only if $f_1(x)$ and
$f_0(x)+f_1(x)$ are both binary bent or semi-bent, respectively,
according as $n$ is even or odd, respectively;

\noindent (3) $\phi(f)(y,z)$ is binary bent or semi-bent,
respectively, according as $n$ is odd or even, respectively, if
$f(x)$ is quaternary bent.
\end{theorem}

\begin{remark}
The binary Boolean function $\phi(f)$ in Theorem \ref{connection} is
often called the Gray image of the quaternary Boolean function $f$.
\end{remark}

\subsection{Galois ring of characteristic 4}

A Galois ring is a Galois extension of an integer residue ring with
a prime power moduli, and this prime power is called the
characteristic of it. Since we will only focus on the quaternary
case, we just recall some basic results of Galois rings of
characteristic 4. For their proofs, we refer to \cite{mcdonald}.

The Galois ring $GR(4,n)$ with $4^n$ elements is an $n$-th Galois
extension of $\mathbb{Z}_4$. In fact, the Galois theory of ring
extensions are much like that of field extensions. More precisely,
to obtain $GR(4,n)$, we can add (formal) roots of a monic basic
irreducible polynomial of degree $n$ over $\mathbb{Z}_4$ to
$\mathbb{Z}_4$. Here a basic irreducible polynomial over
$\mathbb{Z}_4$ is a polynomial whose modulo 2 reduction is a
primitive polynomial over $\mathbb{F}_2$. Assume $\xi$ is a root of
this monic basic irreducible polynomial of order $2^n-1$ (i.e.
$\xi^{2^n-1}=1$), then $GR(4,n)\cong \mathbb{Z}_4[\xi]$. The set
$\mathbb{T}=\{0,1,\xi,\ldots,\xi^{2^n-2}\}$ is called the
Teichm\"uller set of $GR(4,n)$. It is just the roots set of the
equation $x^{2^n}-x=0$ in $GR(4,n)$. It is obvious that
$\mathbb{T}^*=\mathbb{T}\backslash\{0\}$ forms a cyclic group under
the multiplication of $GR(4,n)$, which is isomorphic to
$\mathbb{F}_{2^n}^*$, the multiplicative group of the finite field
$\mathbb{F}_{2^n}$. However, $\mathbb{T}$ does not form a group
under the addition of $GR(4,n)$ since it is not closed under this
operation. To make $\mathbb{T}$ into an additive group, we can
introduce a new operation ``$\oplus$" defined by
\[x\oplus y=x+y+2\sqrt{xy}\]
for any $x,~y\in\mathbb{T}$ (here $\sqrt{xy}$ denotes
$(xy)^{2^{n-1}}$). It can be proved that under the multiplication of
$GR(4,n)$ and the addition ``$\oplus$", $\mathbb{T}$ forms a field
which is isomorphic to the finite field $\mathbb{F}_{2^n}$. Besides,
if we denote by $\mu$ the modulo 2 reduction map, we have
$\mu(\mathbb{T})=\mathbb{F}_{2^n}$.

Every element  $z\in GR(4,n)$ can be uniquely represented as
$z=x+2y$ where $x,~y\in\mathbb{T}$. Under this representation, the
trace function from $GR(4,n)$ to $\mathbb{Z}_4$ is defined as
\[\Tr_1^n(z)=\sum_{i=0}^{n-1}(x^{2^i}+y^{2^i}).\]
The trace function over $GR(4,n)$ and that over $\mathbb{F}_{2^n}$
are related via the map $\mu$ as follows:

\noindent(1) $\mu\left(\Tr_1^n(z)\right)=\tr_1^n(\mu(z))$;\\
(2) $\Tr_1^n(2z)=2\Tr_1^n(z)=2\tr_1^n(\mu(z))$.

\subsection{Quaternary quadratic form}

A quaternary quadratic form $F$ is a mapping from $\mathbb{T}$ to
$\mathbb{Z}_4$ satisfying $F(0)=0$ and \[F(x\oplus
y)=F(x)+F(y)+2B(x,y),\] where $B$ is a symmetric bilinear form on
$\mathbb{T}$, i.e. $B$ can induce a map $\mathbb{T}\times
\mathbb{T}\longrightarrow \mathbb{Z}_4$ such that $B(x,y)=B(y,x)$
and $B(x\oplus y,z)=B(x,z)+B(y,z)$. $B$ is often called the
associate bilinear form of $F$. The rank of $F$ is defined as the
codimension of the radical space of $B$, $\rad(B)$, over
$\mathbb{F}_2$, where
\[\rad(B)=\{x\in\mathbb{T}\mid B(x,y)=0~\text{for~any}~y\in\mathbb{T}\}\]
(it is easy to see that $\rad(B)$ is a vector space over $\mathbb{
F}_2$).

Clearly a quaternary quadratic form can be viewed as a quaternary
Boolean function. Direct computations show that its Fourier
transform can be completely determined by its rank. Details of
computing exponential sums over Galois rings can be found in
\cite{schmidt3,schmidt4}. We just involve a main result here.

\begin{theorem}[\cite{schmidt4}]\label{gbentcrit}
A quaternary quadratic form $F$ is quaternary bent if and only if it
is of full rank, or equivalently, $\hbox{\rm{rad}}(B)=\{0\}$ where
$B$ is the associate bilinear form of $F$.
\end{theorem}

\section{{New constructions of quaternary bent functions}}\label{secZ4bent}

In the rest part of the paper, we assume $n=em$, $e\geq1$ and fix a
positive integer $k$. Denote by $\mathbb{T}_e^*$ the set of nonzero
elements in $\mathbb{T}$ satisfying $x^{2^e}=x$, where $\mathbb{T}$
is the Teichm\"uller set of $GR(4,n)$. In the following, we study
bentness of a special quaternary quadratic form in $n$ variables.

\begin{theorem}\label{mainthm}
Assume $\alpha,~\beta\in\mathbb{T}_e^*$ and $\beta=\alpha^2$. Let
\begin{equation}\label{wuconst}
Q(x)=\Tr^n_1\left(\alpha
x+2\sum_{i=1}^{\lfloor\frac{m-1}{2}\rfloor}c_i\beta
x^{1+2^{eki}}\right),~x\in\mathbb{T},
\end{equation}
where  $c_i\in\mathbb{Z}_2$, $1\leq i\leq
\lfloor\frac{m-1}{2}\rfloor$. Then $Q(x)$ is  a quaternary bent
function if and only if $\gcd(c(x^k),x^m-1)=1$ where
\begin{equation}\label{cx}
c(x)=1+\sum_{i=1}^{\lfloor\frac{m-1}{2}\rfloor}c_i(x^i+x^{m-i})\in\mathbb{F}_2[x].
\end{equation}
\end{theorem}

It is obvious that when $e=1$ and $\alpha=1$, Theorem \ref{mainthm}
coincides with  \cite[Theorem~1]{li}. Hence the construction of
quaternary bent functions in Theorem \ref{mainthm} can be viewed as
a generalization of that in \cite{li} as long as $n$ is not a prime.

To prove Theorem \ref{mainthm}, we need the following lemmas.

\begin{lemma}[\cite{lidl}]\label{lemLP}
Let $L(x)=\sum_{i=0}^{n-1}c_ix^{2^i}$ and
$l(x)=\sum_{i=0}^{n-1}c_ix^{i}$  both be polynomials over
$\mathbb{F}_2$. Then $L(x)$ has only one root  in $\mathbb{F}_{2^n}$
if and only if $\gcd(l(x),x^n-1)=1$.
\end{lemma}

\begin{lemma}\label{lemgcd}
Let $p(x),~q(x)\in\mathbb{F}_2[x]$ and $s$ be any fixed positive
integer. Then $\gcd(p(x),q(x))=1$ if and only if
$\gcd(p(x^s),q(x^s))=1$.
\end{lemma}
\begin{proof}
If $\gcd(p(x),q(x))=1$, there exist $a(x)$ and $b(x)$ over
$\mathbb{F}_2$ such that $a(x)p(x)+b(x)q(x)=1$ according to
B\'ezout's identity. Then we have $a(x^s)p(x^s)+b(x^s)q(x^s)=1$,
which implies $\gcd(p(x^s),q(x^s))=1$. On the contrary, assume
$\gcd(p(x^s),q(x^s))=1$. If $\gcd(p(x),q(x))=d(x)$ with $\deg d>0$,
there exist $a(x)$ and $b(x)$ over $\mathbb{F}_2$ such that
$a(x)p(x)+b(x)q(x)=d(x)$ and thus
$a(x^s)p(x^s)+b(x^s)q(x^s)=d(x^s)$, which contradicts the fact
$\gcd(p(x^s),q(x^s))=1$. \qedd
\end{proof}

\noindent \textit{Proof of Theorem \ref{mainthm}} ~~Obviously $Q(x)$
is a quaternary quadratic form. The bilinear form $B$ associated to
it satisfies
\begin{eqnarray*}
  2B(x,y) &=&Q(x\oplus y)-Q(x)-Q(y) \\
   &=&\Tr^n_1\left(\alpha
(x+y+2\sqrt{xy})+2\sum_{i=1}^{\lfloor\frac{m-1}{2}\rfloor}c_i\beta
(x+y+2\sqrt{xy})^{1+2^{eki}}\right)\\&&-\Tr^n_1\left(\alpha
x+2\sum_{i=1}^{\lfloor\frac{m-1}{2}\rfloor}c_i\beta
x^{1+2^{eki}}\right)-\Tr^n_1\left(\alpha
y+2\sum_{i=1}^{\lfloor\frac{m-1}{2}\rfloor}c_i\beta
y^{1+2^{eki}}\right)\\
   &=&2\tr_1^n\left(\overline{\alpha\sqrt{xy}}\right)+2\tr_1^n\left(\sum_{i=1}^{\lfloor\frac{m-1}{2}\rfloor}c_i
   \bar{\beta}(\bar{x}+\bar{y})^{1+2^{eki}}\right)\\
   &&-2\tr_1^n\left(\sum_{i=1}^{\lfloor\frac{m-1}{2}\rfloor}c_i
   \bar{\beta}\bar{x}^{1+2^{eki}}\right)-2\tr_1^n\left(\sum_{i=1}^{\lfloor\frac{m-1}{2}\rfloor}c_i
   \bar{\beta}\bar{y}^{1+2^{eki}}\right)\\
&=&2\tr_1^n\left(\bar{\alpha}^2\bar{x}\bar{y}\right)+2\tr_1^n\left(\sum_{i=1}^{\lfloor\frac{m-1}{2}\rfloor}c_i
   \bar{\beta}\left(\bar{x}\bar{y}^{2^{eki}}+\bar{x}^{2^{eki}}\bar{y}\right)\right)
\end{eqnarray*}
(here we distinguish $\bar{x}$ with $\mu(x)$ for any $x\in
\mathbb{T}$ for simplicity).  Since $\bar{\alpha}^2=\bar{\beta}$ and
for $1\leq i\leq \lfloor\frac{m-1}{2}\rfloor$,
\[  \tr_1^n\left(c_i\bar{\beta}\bar{x}\bar{y}^{2^{eki}}\right) =
  \tr_1^n\left((c_i\bar{\beta}\bar{x})^{2^{n-eki}}\bar{y}\right)
   =\tr_1^n\left(c_i\bar{\beta}\bar{x}^{2^{ek(m-i)}}\bar{y}\right)\]
as $\beta\in\mathbb{T}_e^*$, we have
\[2B(x,y)=2\tr_1^n\left(\bar{\beta}\bar{y}\left[\bar{x}+\sum_{i=1}^{\lfloor\frac{m-1}{2}\rfloor}c_i\left(
\bar{x}^{2^{eki}}+ \bar{x}^{2^{ek(m-i)}}\right) \right]\right).\]
Hence from Theorem \ref{gbentcrit} we know that $Q(x)$ is quaternary
bent if and only if the polynomial
\[L(x)=x+\sum_{i=1}^{\lfloor\frac{m-1}{2}\rfloor}c_i\left(
{x}^{2^{eki}}+ {x}^{2^{ek(m-i)}}\right) \] has only one root $0$ in
$\mathbb{F}_{2^n}$. By Lemma \ref{lemLP}, this is equivalent to
\[\gcd\left(1+\sum_{i=1}^{\lfloor\frac{m-1}{2}\rfloor}c_i\left(
{x}^{{eki}}+
{x}^{{ek(m-i)}}\right),~x^n-1\right)=\gcd\left(1+\sum_{i=1}^{\lfloor\frac{m-1}{2}\rfloor}c_i\left(
{x}^{{eki}}+ {x}^{{ek(m-i)}}\right),~x^{em}-1\right)=1,\] which is
further equivalent to
\[\gcd\left(1+\sum_{i=1}^{\lfloor\frac{m-1}{2}\rfloor}c_i\left(
{x}^{{ki}}+ {x}^{{k(m-i)}}\right),~x^{m}-1\right)=1\] by Lemma
\ref{lemgcd}. \qedd \medskip

We call a set of elements  $\mathcal{C}=\{c_i\in\mathbb{F}_2\mid
1\leq i\leq \lfloor\frac{m-1}{2}\rfloor\}$ a QBF-set w.r.t $(n,k)$
if the polynomial $c(x)$ in the form \eqref{cx} defined by it
satisfies $\gcd(c(x^k),x^m-1)=1$. By Theorem \ref{mainthm}, explicit
constructions of QBF-sets will promise new classes of  quaternary
bent functions
 in the form \eqref{wuconst}. In the remaining part of this section, we devote to finding several
 constructions of QBF-sets.

Firstly, from the constructions of quaternary bent functions
proposed in \cite{li}, we can directly obtain some classes of
QBF-sets. The results, whose detailed proofs can be found in
\cite[Corollary~1,2,3,4,5]{li}, are summarized in the following
proposition.

\begin{proposition}\label{egknown}
(1) Assume $m\geq7$ and let $\mathcal{C}=\{2,3\}$. Then
$\mathcal{C}$ is a QBF-set w.r.t. $(n,k)$ if and only if
$\gcd(m,3k)=\gcd(m,k)$;

(2) Let $t$ be a positive integer with $t<(m+1)/4$. Let
$\mathcal{C}=\{2i+1\mid 0\leq i\leq t\}$. Then $\mathcal{C}$ is a
QBF-set w.r.t. $(n,k)$ if and only if
$\gcd(m,(2t+3)k)=\gcd(m,(2t+1)k)=\gcd(m,k)$;

(3) Let $t$ be a positive integer with $t<(m+3)/4$. Let
$\mathcal{C}=\{2i\mid 1\leq i\leq t\}\cup\{1\}$. Then $\mathcal{C}$
is a QBF-set w.r.t. $(n,k)$ if and only if
$\gcd(m,(2t+3)k)=\gcd(m,(2t-1)k)=\gcd(m,k)$;

(4) Let $t$ be a positive integer with $t<m/2$. Let
$\mathcal{C}=\{i\mid 1\leq i\leq t\}$. Then $\mathcal{C}$ is a
QBF-set w.r.t. $(n,k)$ if and only if $\gcd(m,(2t+1)k)=\gcd(m,k)$;

(5) Assume $m\geq13$ and let $\mathcal{C}=\{2,5,6\}$. Then
$\mathcal{C}$ is a QBF-set w.r.t. $(n,k)$ if and only if
$\gcd(m,5k)=\gcd(m,k)$;

(6) Assume $m\geq11$ and let $\mathcal{C}=\{1,4,5\}$. Then
$\mathcal{C}$ is a QBF-set w.r.t. $(n,k)$ if and only if
$\gcd(m,3k)=\gcd(m,7k)=\gcd(m,k)$;

(7) Assume $m\geq13$ and let $\mathcal{C}=\{3,5,6\}$. Then
$\mathcal{C}$ is a QBF-set w.r.t. $(n,k)$ if and only if
$\gcd(m,3k)=\gcd(m,5k)=\gcd(m,7k)=\gcd(m,k)$.
\end{proposition}

In addition, we can construct some new classes of QBF-sets.

\begin{proposition}\label{eg1}
Let $t,~s$ be positive integers with
$s<t<\lfloor\frac{m-1}{2}\rfloor $. Let
$\mathcal{C}=\{t-s,s,t,t+s\}$. Then $\mathcal{C}$ is a QBF-set
w.r.t. $(n,k)$ if and only if $\gcd(m,3tk)=\gcd(m,tk)$ and
$\gcd(m,3sk)=\gcd(m,sk)$.
\end{proposition}
\begin{proof}
Since
\[c(x)=1+x^{t-s} +x^s+x^t+x^{t+s}+x^{m-(t-s)}+x^{m-s}+x^{m-t}+x^{m-(t+s)},\]
we have
\begin{eqnarray*}
  \gcd(c(x^k),x^m-1) &=&\gcd(1+x^{k(t-s)} +x^{ks}+x^{kt}+x^{k(t+s)}+x^{m-k(t-s)}+x^{m-ks}
  +x^{m-kt}+x^{m-k(t+s)},x^m-1)  \\
&=&\gcd(x^{k(t+s)}+x^{2kt}
+x^{kt+2ks}+x^{2kt+ks}+x^{2k(t+s)}+x^{2ks}+x^{k(t-s)}
  +x^{ks}+1,x^m-1)\\
   &=&\gcd((1+x^{kt}+x^{2kt})(1+x^{ks}+x^{2ks}),x^m-1)\\
   &=&\gcd\left(\frac{x^{3kt}-1}{x^{kt}-1}\cdot\frac{x^{3ks}-1}{x^{ks}-1},~x^m-1\right).
\end{eqnarray*}
Thus $\gcd(c(x^k),x^m-1)=1$ if and only if $\gcd(m,3tk)=\gcd(m,tk)$
and $\gcd(m,3sk)=\gcd(m,sk)$.\qedd
\end{proof}

\begin{proposition}\label{eg2}
Assume $m\geq11$ and let $\mathcal{C}=\{1,3,4,5\}$. Then
$\mathcal{C}$ is a QBF-set w.r.t. $(n,k)$ if and only if
$\gcd(m,3k)=\gcd(m,k)$.
\end{proposition}
\begin{proof}
Since
\[c(x)=1+x+x^3+x^4+x^5+x^{m-5}+x^{m-4}+x^{m-3}+x^{m-1},\]
we have
\begin{eqnarray*}
  \gcd(c(x^k),x^m-1)
  &=&\gcd(1+x^k+x^{3k}+x^{4k}+x^{5k}+x^{m-5k}+x^{m-4k}+x^{m-3k}+x^{m-k},x^m-1)\\
  &=&\gcd(x^{5k}+x^{6k}+x^{8k}+x^{9k}+x^{10k}+1+x^{k}+x^{2k}+x^{4k},x^m-1)\\
  &=&\gcd((1+x^k+x^{2k})^5,x^m-1)\\
  &=&\gcd\left(\left(\frac{x^{3k}-1}{x^k-1}\right)^5,~x^m-1\right).
\end{eqnarray*}
Thus $\gcd(c(x^k),x^m-1)=1$ if and only if
$\gcd(m,3k)=\gcd(m,k)$.\qedd
\end{proof}

\begin{proposition}\label{eg3}
Assume $m\geq13$ and let $\mathcal{C}=\{1,2,5,6\}$. Then
$\mathcal{C}$ is a QBF-set w.r.t. $(n,k)$ if and only if
$\gcd(m,3k)=\gcd(m,k)$.
\end{proposition}
\begin{proof}
Since
\[c(x)=1+x+x^2+x^5+x^6+x^{m-6}+x^{m-5}+x^{m-2}+x^{m-1},\]
we have
\begin{eqnarray*}
  \gcd(c(x^k),x^m-1)
  &=&\gcd(1+x^k+x^{2k}+x^{5k}+x^{6k}+x^{m-6k}+x^{m-5k}+x^{m-2k}+x^{m-k},x^m-1)\\
  &=&\gcd(x^{6k}+x^{7k}+x^{8k}+x^{11k}+x^{12k}+1+x^{k}+x^{4k}+x^{5k},x^m-1)\\
  &=&\gcd((1+x^{3k}+x^{6k})(1+x^k+x^{2k})^3,x^m-1)\\
  &=&\gcd\left(\frac{x^{9k}-1}{x^{3k}-1}\cdot\left(\frac{x^{3k}-1}{x^k-1}\right)^3,~x^m-1\right).
\end{eqnarray*}
Thus $\gcd(c(x^k),x^m-1)=1$ if and only if $\gcd(m,3k)=\gcd(m,k)$
and $\gcd(m,9k)=\gcd(m,3k)$. But $\gcd(m,3k)=\gcd(m,k)$ implies that
$\gcd(m,9k)=\gcd(m,3k)$, so this condition is enough.\qedd
\end{proof}

\begin{proposition}\label{eg4}
Assume $m\geq13$ and let $\mathcal{C}=\{2,3,6\}$. Then $\mathcal{C}$
is a QBF-set w.r.t. $(n,k)$ if and only if
$\gcd(m,3k)=\gcd(m,7k)=\gcd(m,k)$.
\end{proposition}
\begin{proof}
Since
\[c(x)=1+x^2+x^3+x^6+x^{m-6}+x^{m-3}+x^{m-2},\]
we have
\begin{eqnarray*}
  \gcd(c(x^k),x^m-1)
  &=&\gcd(1+x^{2k}+x^{3k}+x^{6k}+x^{m-6k}+x^{m-3k}+x^{m-2k},x^m-1)\\
  &=&\gcd(x^{6k}+x^{8k}+x^{9k}+x^{12k}+1+x^{3k}+x^{4k},x^m-1)\\
  &=&\gcd((1+x^{k}+x^{2k})^3(1+x^k+x^{2k}+x^{3k}+x^{4k}+x^{5k}+x^{6k}),x^m-1)\\
  &=&\gcd\left(\left(\frac{x^{3k}-1}{x^{k}-1}\right)^3\cdot\frac{x^{7k}-1}{x^k-1},~x^m-1\right).
\end{eqnarray*}
Thus $\gcd(c(x^k),x^m-1)=1$ if and only if
$\gcd(m,3k)=\gcd(m,7k)=\gcd(m,k)$.\qedd
\end{proof}

\begin{proposition}\label{eg5}
Assume $m\geq15$ and let $\mathcal{C}=\{2,3,4,7\}$. Then
$\mathcal{C}$ is a QBF-set w.r.t. $(n,k)$ if and only if
$\gcd(m,3k)=\gcd(m,5k)=\gcd(m,k)$.
\end{proposition}
\begin{proof}
Since
\[c(x)=1+x^2+x^3+x^4+x^7+x^{m-7}+x^{m-4}+x^{m-3}+x^{m-2},\]
we have
\begin{eqnarray*}
  \gcd(c(x^k),x^m-1)
  &=&\gcd(1+x^{2k}+x^{3k}+x^{4k}+x^{7k}+x^{m-7k}+x^{m-4k}+x^{m-3k}+x^{m-2k},x^m-1)\\
  &=&\gcd(x^{7k}+x^{9k}+x^{10k}+x^{11k}+x^{14k}+1+x^{3k}+x^{4k}+x^{5k},x^m-1)\\
  &=&\gcd((1+x^{k}+x^{2k})(1+x^k+x^{2k}+x^{3k}+x^{4k})^3,x^m-1)\\
  &=&\gcd\left(\frac{x^{3k}-1}{x^{k}-1}\cdot\left(\frac{x^{5k}-1}{x^k-1}\right)^3,~x^m-1\right).
\end{eqnarray*}
Thus $\gcd(c(x^k),x^m-1)=1$ if and only if
$\gcd(m,3k)=\gcd(m,5k)=\gcd(m,k)$.\qedd
\end{proof}

\begin{proposition}\label{eg6}
Assume $m\geq15$ and let $\mathcal{C}=\{1,2,3,7\}$. Then
$\mathcal{C}$ is a QBF-set w.r.t. $(n,k)$ if and only if
$\gcd(m,3k)=\gcd(m,5k)=\gcd(m,7k)=\gcd(m,k)$.
\end{proposition}
\begin{proof}
Since
\[c(x)=1+x+x^2+x^3+x^7+x^{m-7}+x^{m-3}+x^{m-2}+x^{m-1},\]
we have
\begin{eqnarray*}
  \gcd(c(x^k),x^m-1)
  &=&\gcd(1+x^{k}+x^{2k}+x^{3k}+x^{7k}+x^{m-7k}+x^{m-3k}+x^{m-2k}+x^{m-k},x^m-1)\\
  &=&\gcd(x^{7k}+x^{8k}+x^{9k}+x^{10k}+x^{14k}+1+x^{4k}+x^{5k}+x^{6k},x^m-1)\\
  &=&\gcd((1+x^{k}+x^{2k})^2(1+x^k+x^{2k}+x^{3k}+x^{4k})(1+x^k+x^{2k}+x^{3k}+x^{4k}+x^{5k}+x^{6k}),x^m-1)\\
  &=&\gcd\left(\left(\frac{x^{3k}-1}{x^{k}-1}\right)^2\cdot\frac{x^{5k}-1}{x^k-1}\cdot\frac{x^{7k}-1}{x^k-1},~x^m-1\right).
\end{eqnarray*}
Thus $\gcd(c(x^k),x^m-1)=1$ if and only if
$\gcd(m,3k)=\gcd(m,5k)=\gcd(m,7k)=\gcd(m,k)$.\qedd
\end{proof}

\begin{proposition}\label{eg7}
Assume $m\geq17$ and let $\mathcal{C}=\{1,3,5,6,8\}$. Then
$\mathcal{C}$ is a QBF-set w.r.t. $(n,k)$ if and only if
$\gcd(m,3k)=\gcd(m,7k)=\gcd(m,k)$.
\end{proposition}
\begin{proof}
Since
\[c(x)=1+x+x^3+x^5+x^6+x^8+x^{m-8}+x^{m-6}+x^{m-5}+x^{m-3}+x^{m-1},\]
we have
\begin{eqnarray*}
  \gcd(c(x^k),x^m-1)
  &=&\gcd(1+x^{k}+x^{3k}+x^{5k}+x^{6k}+x^{8k}+x^{m-8k}+x^{m-6k}+x^{m-5k}+x^{m-3k}+x^{m-k},x^m-1)\\
  &=&\gcd(x^{8k}+x^{9k}+x^{11k}+x^{13k}+x^{14k}+x^{16k}+1+x^{2k}+x^{3k}+x^{5k}+x^{6k},x^m-1)\\
  &=&\gcd((1+x^{k}+x^{2k})^5(1+x^k+x^{2k}+x^{3k}+x^{4k}+x^{5k}+x^{6k}),x^m-1)\\
  &=&\gcd\left(\left(\frac{x^{3k}-1}{x^{k}-1}\right)^5\cdot\frac{x^{7k}-1}{x^k-1},~x^m-1\right).
\end{eqnarray*}
Thus $\gcd(c(x^k),x^m-1)=1$ if and only if
$\gcd(m,3k)=\gcd(m,7k)=\gcd(m,k)$.\qedd
\end{proof}

\begin{proposition}\label{eg8}
Let $t$ be a positive integer with $t<(m-3)/4$. Let
$\mathcal{C}=\{2i+1\mid 1\leq i\leq t\}\cup\{2t+2\}$. Then
$\mathcal{C}$ is a QBF-set w.r.t. $(n,k)$ if and only if
$\gcd(m,(2t+3)k)=\gcd(m,(2t+1)k)=\gcd(m,3k)=\gcd(m,k)$.
\end{proposition}
\begin{proof}
Since
\[c(x)=1+\sum_{i=1}^{t}(x^{2i+1}+x^{m-(2i+1)})+x^{2t+2}+x^{m-(2t+2)},\]
we have
\begin{eqnarray*}
  \gcd(c(x^k),x^m-1)
  &=&\gcd\left(1+\sum_{i=1}^{t}(x^{(2i+1)k}+x^{m-(2i+1)k})+x^{(2t+2)k}+x^{m-(2t+2)k},x^m-1\right)\\
  &=&\gcd\left(x^{(2t+2)k}+x^{(2t+2)k}\sum_{i=1}^{t}x^{(2i+1)k}+\sum_{i=1}^{t}x^{(2i-1)k}+x^{(4t+4)k}+1,x^m-1\right)\\
  &=&\gcd\left((1+x^k+x^{2k})\left(1+\sum_{i=1}^{2t}x^{ik}\right)\left(\sum_{i=1}^{2t+2}x^{ik}\right),x^m-1\right)\\
  &=&\gcd\left(\frac{x^{3k}-1}{x^{k}-1}\cdot\frac{x^{(2t+1)k}-1}{x^{k}-1}\cdot\frac{x^{(2t+3)k}-1}{x^{k}-1},~x^m-1\right).
\end{eqnarray*}
Thus $\gcd(c(x^k),x^m-1)=1$ if and only if
$\gcd(m,(2t+3)k)=\gcd(m,(2t+1)k)=\gcd(m,3k)=\gcd(m,k)$.\qedd
\end{proof}

\begin{proposition}\label{eg9}
Let $t$ be a positive integer with $t<(m-1)/4$. Let
$\mathcal{C}=\{2i\mid 2\leq i\leq t\}\cup\{1,2t+1\}$. Then
$\mathcal{C}$ is a QBF-set w.r.t. $(n,k)$ if and only if
$\gcd(m,(2t+3)k)=\gcd(m,(2t-1)k)=\gcd(m,3k)=\gcd(m,k)$.
\end{proposition}
\begin{proof}
Since
\[c(x)=1+x+\sum_{i=2}^{t}(x^{2i}+x^{m-2i})+x^{2t+1}+x^{m-(2t+1)}+x^{m-1},\]
we have
\begin{eqnarray*}
  \gcd(c(x^k),x^m-1)
  &=&\gcd\left(1+x^k+\sum_{i=2}^{t}(x^{2ik}+x^{m-2ik})+x^{(2t+1)k}+x^{m-(2t+1)k}+x^{m-k},x^m-1\right)\\
  &=&\gcd\left(x^{(2t+1)k}+x^{(2t+2)k}+x^{(2t+1)k}\sum_{i=2}^{t}x^{2ik}+\sum_{i=1}^{t-1}x^{(2i-1)k}+x^{(4t+2)k}+1+x^{2tk},x^m-1\right)\\
  &=&\gcd\left((1+x^k+x^{2k})\left(1+\sum_{i=1}^{2t-2}x^{ik}\right)\left(\sum_{i=1}^{2t+2}x^{ik}\right),x^m-1\right)\\
  &=&\gcd\left(\frac{x^{3k}-1}{x^{k}-1}\cdot\frac{x^{(2t-1)k}-1}{x^{k}-1}\cdot\frac{x^{(2t+3)k}-1}{x^{k}-1},~x^m-1\right).
\end{eqnarray*}
Thus $\gcd(c(x^k),x^m-1)=1$ if and only if
$\gcd(m,(2t+3)k)=\gcd(m,(2t-1)k)=\gcd(m,3k)=\gcd(m,k)$.\qedd
\end{proof}

By virtue of the QBF-sets characterized in Proposition 1 to 10, we
can obtain many new classes of quaternary bent functions.

\section{New classes of quadratic bent and semi-bent function in polynomial forms}\label{secbent}

According to the connections between quaternary bent functions and
binary bent and semi-bent functions indicated in Theorem
\ref{connection}, we can derive some new classes of binary bent and
semi-bent functions based on the classes of quaternary bent
functions constructed in Section \ref{secZ4bent}.

Firstly, we recall a result concerning 2-adic expansion of the trace
function over $GR(4,n)$.

\begin{lemma}[\cite{hammous}]\label{lemtrace}
\[\Tr_1^n(x)=\tr_1^n(\bar{x})+2p(\bar{x}),\]where $p(x)$ is defined by
\eqref{eqpx}.
\end{lemma}

From Lemma \ref{lemtrace}, we know that $\Tr_1^n(\lambda
x)=\tr_1^n(\bar{\lambda}\bar{x})+2p(\bar{\lambda}\bar{x})$ for any
$\lambda\in\mathbb{T}$.  Hence the quaternary Boolean function
defined by \eqref{wuconst} can be expressed as
\[Q(x)=\tr_1^n(\bar{\alpha}\bar{x})+2\left[p(\bar{\alpha}\bar{x})+\sum_{i=1}^{\lfloor\frac{m-1}{2}\rfloor}c_i\tr_1^n\left(\bar{\beta}
\bar{x}^{1+2^{eki}}\right)\right].\] From Theorem \ref{connection}
(2), we obtain the following result, which generalizes
\cite[Theorem~3]{li}.

\begin{theorem}\label{derivethm}
Assume $\alpha,~\beta\in\mathbb{F}_{2^e}^*$ and $\beta=\alpha^2$.
Let
\begin{equation}\label{wuconstderive}
f_Q(x)=p(\alpha{x})+\sum_{i=1}^{\lfloor\frac{m-1}{2}\rfloor}c_i\tr_1^n\left(\beta
{x}^{1+2^{eki}}\right),~x\in\mathbb{F}_{2^n},
\end{equation}
where  $c_i\in\mathbb{Z}_2$, $1\leq i\leq
\lfloor\frac{m-1}{2}\rfloor$. Then $f_Q(x)$ is  a bent function or a
semi-bent function, respectively, according as $n$ is even or odd,
respectively, if and only if $\gcd(c(x^k),x^m-1)=1$ where $c(x)$ is
defined by \eqref{cx}.
\end{theorem}

It is clear that the Boolean function $f_Q(x)$ in Theorem
\ref{derivethm} is a quadratic Boolean function in the so-called
polynomial form \cite{charpin}. In fact, constructions of quadratic
bent and semi-bent functions in polynomial forms are extensively
studied by several authors (see, for example,
\cite{khoo,charpin,yuny,huhg}). From Theorem \ref{derivethm}, we can
directly  obtain many new classes of quadratic bent and semi-bent
functions in polynomial forms by choosing the sets of coefficients
$\{c_i\}$ to be  QBF-sets characterized  by, say, Proposition
\ref{egknown} to \ref{eg9}. However, it is not a simple matter to
get them via studying quadratic forms over finite fields, which is
the standard approach to study quadratic binary Boolean functions.

\section{Conclusions}\label{secconclud}

In this paper, we propose a new construction of quaternary bent
functions from quadratic  forms over Galois rings of characteristic
4, which generalize some previous work. By characterizing the
so-called QBF-sets, we can explicitly construct several new classes
of quaternary bent functions. They furthermore derive several new
classes of binary bent and semi-bent functions which are quadratic
ones in polynomials forms.



\begin{thebibliography}{50}

\bibitem{charpin}
 Charpin P.,  Pasalic E.,  Tavernier C.: On bent and semi-bent Boolean
functions. IEEE Trans. Inform. Theory 51,  4286--4298 (2005).

\bibitem{chee}
Chee S., Lee S., Kim K.: Semi-bent functions. In Pieprzyk J.,
Safavi-Naini R. (eds.) Asicacrypt 1994. LNCS, vol. 917, pp.
386--397. Springer-Verlag, Berlin (1994).

\bibitem{hammous}
Hammous R., Kumar P., Calderbank A., et al.: The
$\mathbb{Z}_4$-linearity of Kerdock, Preparata, Goethals and related
codes. IEEE Trans. Inform. Theory 40, 301--319 (1994).

\bibitem{huhg}
 Hu H.,  Feng D.: On quadratic bent functions in polynomial forms.
IEEE Trans. Inform. Theory 53, 2610--2615 (2007).

\bibitem{jones}
Jones A., Wilkinson T.: Combined coding for error control and
increased robustness to system nonlinearities in OFDM. In: VTC 1996.
Proc. IEEE 46th Veh. Technol. Conf., pp. 904--908. IEEE Computer
Society, Washington, D.C. (1996).

\bibitem{khoo}
 Khoo K.,  Gong G.,  Stinson D.:  A new famlily of Gold-like
 sequences.
In: ISIT 2002. Proc. IEEE Int. Symp. Inform. Theory, p. 181. IEEE
Computer Society, Washington, D.C. (2002).

\bibitem{li}
Li N., Tang X., Helleseth T.: New classes of generalized Boolean
bent functions over $\mathbb{Z}_4$. In: ISIT 2012. Proc. IEEE Int.
Symp. Inform. Theory, pp. 841--845. IEEE Computer Society,
Washington, D.C. (2012).

\bibitem{lidl}
 Lidl R.,  Niederreiter H.: Finite fields.  Cambridge University Press, Cambridge (1997).


\bibitem{mcdonald}
McDonald B.: Finite rings with identity. Dekker, New York (1974).


\bibitem{rothaus}
 Rothaus O.: On bent functions. J. Combin. Theory Ser. A 20, 300--305 (1976).


\bibitem{schmidt2}
Schmidt K.-U.: On spectrally-bounded codes for multicarrier
communications. Ph.D thesis. Dresden University of Technology,
Dresden (2007).

\bibitem{schmidt3}
Schmidt K.-U.: $\mathbb{Z}_4$-valued quadratic forms and expinential
sums. In: ISIT 2008. Proc. IEEE Int. Symp. Inform. Theory, pp.
2762--2766. IEEE Computer Society, Washington, D.C. (2008).

\bibitem{schmidt1}
Schmidt K.-U.: Quaternary constant-amplitude codes for multicode
CDMA. IEEE Trans. Inform. Theory 55, 1824--1832 (2009).


\bibitem{schmidt4}
Schmidt K.-U.: $\mathbb{Z}_4$-valued quadratic forms and quaternary
sequence families. IEEE Trans. Inform. Theory 55, 5803--5810 (2009).


\bibitem{stanica}
St\u{a}nic\u{a} P., Martinsen T., Gangopadhyay S., et al.: Bent and
generalized bent Boolean functions. Des. Codes Cryptogr. 69, 77--94
(2013).

\bibitem{yuny}
 Yu N.,  Gong G.: Constructions of quadratic bent functions in
polynomial forms. IEEE Trans. Inform. Theory 53, 3291--3299 (2006).


\end{thebibliography}
\end{document}